\documentclass[journal]{IEEEtran}

\usepackage[noadjust]{cite}

\usepackage{color,xcolor}

\ifCLASSINFOpdf
  \usepackage[pdftex]{graphicx}
  \graphicspath{{/}}
\else
\fi

\usepackage{amsmath,amssymb}
\ifCLASSOPTIONcompsoc
 \usepackage[caption=false,font=normalsize,labelfont=sf,textfont=sf]{subfig}
\else
 \usepackage[caption=false,font=footnotesize]{subfig}
\fi

\usepackage{url}

\usepackage{array}
\usepackage[]{booktabs}
\usepackage[fleqn]{empheq}
\DeclareMathOperator*{\sign}{sign}
\DeclareMathOperator*{\SGN}{SGN}
\DeclareMathOperator*{\sat}{sat}
\newcommand\real{\ensuremath{{\mathbb R}}}
\newcommand\Di{\ensuremath{D_+}}
\newcommand{\smallmat}[1]{\left[ \begin{smallmatrix}#1
    \end{smallmatrix} \right]}
\newcommand{\bigmat}[1]{\begin{bmatrix} #1 \end{bmatrix}}
\newcommand\A{\ensuremath{{\mathcal A}}}
\newcommand\K{\ensuremath{{\mathcal K}}}

\newcommand\M{\ensuremath{{\mathcal M}}}

\newcommand\dz{\ensuremath{\mathrm{dz}_{f_c}}}

\newcommand\ball{{\mathbb B}}
\newtheorem{thm}{Theorem}
\newenvironment{theorem}{\begin{thm}\rm }{\end{thm}}
\newtheorem{remm}{Remark}
\newenvironment{remark}{\begin{remm}\rm }{\hfill \hspace*{1pt} \hfill $\lrcorner$\end{remm}}
\newtheorem{defi}{Definition}

\newtheorem{lem}{Lemma}
\newenvironment{lemma}{\begin{lem}\rm }{\end{lem}}
\newtheorem{fac}{Fact}
\newenvironment{fact}{\begin{fac}\rm }{\end{fac}}
\newtheorem{propo}{Proposition}
\newenvironment{proposition}{\begin{propo}\rm }{\end{propo}}
\newtheorem{cla}{Claim}
\newenvironment{claim}{\begin{cla}\rm }{\end{cla}}
\newtheorem{coroll}{Corollary}
\newenvironment{corollary}{\begin{coroll}\rm }{\end{coroll}}
\newtheorem{ass}{Assumption}
\newenvironment{assumption}{\begin{ass}\rm }{\end{ass}}
\newenvironment{proof}{\noindent {\em Proof.}}{\hfill \hspace*{1pt}\hfill $\square$\\}
\newenvironment{proofof}{\noindent {\em Proof of }}{\hfill \hspace*{1pt}\hfill $\square$\\}

\begin{document}
\title{Global asymptotic stability\\ of a PID control system with Coulomb friction}
\author{Andrea~Bisoffi, Mauro~Da~Lio, Andrew~R.~Teel and~Luca~Zaccarian% <-this % stops a space
\thanks{Research supported in part by AFOSR grant FA9550-15-1-0155 and NSF grant ECCS-1508757, by the ANR project LimICoS contract number 12 BS03 005 01, and by grant OptHySYS funded by the University of Trento.}% <-this % stops a space
\thanks{A. Bisoffi and M. Da Lio are with Dipartimento di Ingegneria Industriale, University of Trento, Italy \texttt{\{andrea.bisoffi, mauro.dalio\}@unitn.it}}% <-this % stops a space
\thanks{A. R. Teel is with the Department of Electrical and Computer Engineering, University of California, Santa Barbara, CA 93106, USA \texttt{teel@ece.ucsb.edu}}% <-this % stops a space
\thanks{L. Zaccarian is with CNRS, LAAS, 7 avenue du Colonel Roche, F-31400 Toulouse, France and Universit\'e de Toulouse, 7 avenue du Colonel Roche, 31077 Toulouse cedex 4, France, and Dipartimento di Ingegneria Industriale, University of Trento, Italy \texttt{zaccarian@laas.fr}}%
}

\maketitle

\begin{abstract}
We propose a model for representing a point mass subject to Coulomb friction in feedback with a PID controller, based on a differential inclusion comprising all the possible magnitudes of static friction during the stick phase. 
For this model we study the set of all equilibria and we establish its global asymptotic stability using a discontinuous Lyapunov-like function, and a suitable LaSalle's invariance principle.
We finally use well-posedness of the proposed model to establish useful robustness results, including an ISS property from a suitable input in a perturbed context. Simulation results are also given to illustrate our statements.
\end{abstract}

\IEEEpeerreviewmaketitle

\section{Introduction}

\IEEEPARstart{F}{riction} in mechanical systems has been investigated since the times of Leonardo da Vinci, Guillaume Amontons, Charles-Augustin de Coulomb and Arthur-Jules Morin. Their main findings acknowledge that, for a moving mass, the friction force is proportional to the normal force through a kinetic coefficient (Coulomb friction) and presents possibly a term proportional to the velocity (viscous friction), whereas at rest the friction force is bounded by the product of the normal force and a static coefficient, generally greater than the kinetic coefficient.

Within the control community, the interest in the dynamical properties of friction had its peak in the 1990's, and the control engineering reasons for this interest are lucidly argued in \cite[\S1]{olsson1998friction}.
These properties have been studied along a modeling direction, where we mention the Dahl model \cite{dahl1968solid}, the LuGre model \cite{canudas1995new,astrom2008revisiting}, the models by Bliman and Sorine \cite{bliman1995easy} and the Leuven model \cite{swevers2000integrated}. The characteristics of all these models are also detailed in~\cite{ferretti2004single}. When a mass moves with steady velocity and the corresponding friction force is measured, there is a small interval of velocities near zero where the friction force decreases before increasing again due to viscous friction and this behaviour is given the name of \emph{Stribeck effect}. Another experimental observation is the distinction between the two motion regimes of presliding and sliding. During presliding, the friction force is better described as a function of the (microsliding) displacement (see \cite{rabinowicz1951nature}), which is intuitively due to the asperity junctions that build up between the surfaces and that behave like stiff springs \cite{ferretti2004single}. After a critical value of displacement (and a break-away force) is reached, the sliding motion can begin. This property that the friction force is only position dependent is called \emph{rate independence} \cite{astrom2008revisiting}, and is to be found in the Dahl model \cite[\S4.1]{olsson1998friction}, and in the models by Bliman and Sorine \cite[\S5.1]{olsson1998friction}. In the latter ones, considering friction as depending only on the path allows using the theory of hysteresis operators \cite{visintin2013differential,krasnoselski1989systems}. On the other hand, rate dependence holds for the LuGre model \cite{canudas1995new}. As a final remark, the LuGre model itself proved to be amenable to theoretical analysis, as \cite{barahanov2000necessary} presents necessary and sufficient conditions for the passivity of its underlying operator from velocity to friction force.

In this work, we propose to characterize Coulomb friction in terms of differential inclusions \cite{aubin1984differential}, and we apply this characterization to the case of a point mass under such a friction force and actuated by a proportional-integral-derivative (PID) controller. This problem is a classical one in the friction literature (together with the point mass on a moving belt) and we will be able to prove the global asymptotic stability of the attractor having zero velocity, zero position and a bounded integral error. The use of a set-valued map for the friction force can be seen as quite natural and is taken into consideration in~\cite{bliman1995easy,vandewouw2004attractivity,putra2007analysis}: in~\cite{vandewouw2004attractivity} it is applied to uncontrolled multi-degree-of-freedom mechanical systems, in~\cite{putra2007analysis} to a PD controlled 1 degree-of-freedom system. The combination of set-valued friction laws and Lyapunov tools is also the subject of~\cite[Chap.~5-6]{leine2007stability}.

To the best of the authors' knowledge, global asymptotic stability has not been proved so far. In particular, it was proved (see \cite[Thm.~1]{armstrong1996pid} and the related works \cite{armstrong1993techRep,armstrong1994pid}) that in our same setting there exists no stick-slip limit cycle (the so-called hunting phenomenon), which is the detrimental signature of a stiction greater than the Coulomb friction.
As an overall achievement, Lyapunov tools applied to a differential-inclusion model enable proving global asymptotic stability of the largest set of equilibria.
Additionally, the established properties and the regularity of our model imply robustness of asymptotic stability. This, in turn, allows us to prove an input-to-state stability (ISS) property for the perturbed dynamics, establishing that more general friction phenomena (including the Stribeck effect) cause a gradual deterioration of the response, in an ISS sense.
We regard this work as a stepping stone to stiction larger than Coulomb and to its description through hybrid systems \cite{goebel2012hybrid}, and to proposing compensation schemes using hybrid friction laws.

The paper is structured as follows. We present the proposed model and the main results in Section~\ref{sec:main}. Then Section~\ref{sec:sims} contains an illustration by simulation of the established properties. The end of the paper contains the Lyapunov-based proof, separated into the proof of global attractivity (Section~\ref{sec:attr}) and of stability (Section~\ref{sec:stab}).

\textit{Notation.} The sign function is defined as: $\sign(x):=1$ if $x>0$, $\sign(0):= 0$, $\sign(x):=-1$ if $x<0$. The saturation function is defined as: $\sat(x):=\sign(x)$ if $|x|>1$, $\sat(x):=x$ if $|x|\le1$. For $c\neq 0$, the function $x\mapsto \text{dz}_c(x)$ is defined as $\text{dz}_c(x):=x-c\sat\big(\frac{x}{c}\big)$. $|x|$ denotes the Euclidean norm of vector $x$. $\langle \cdot, \cdot \rangle$ defines the scalar product between its two vector arguments.

\begin{figure}[ht]
\centering
\includegraphics[width=.8\linewidth]{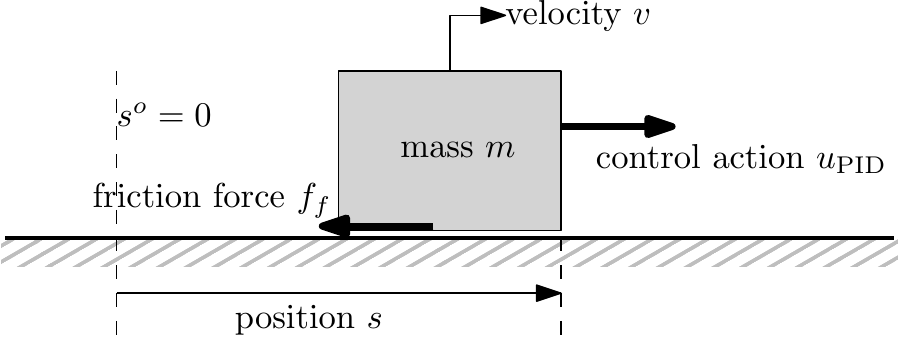}
\caption{Mass under the action of friction and controlled by a PID controller.}
\label{fig:MaPIDscheme}
\end{figure}

\section{Proposed model and main result}
\label{sec:main}

\subsection{Derivation of the model}

Consider a point mass $m$ described by position $s$ and velocity $v$, as in Figure~\ref{fig:MaPIDscheme}. The overall friction force $f_f$ acting on the mass comprises both Coulomb and viscous friction. Its classical description (see \cite[Eq.~(3)]{armstrong1993techRep}, or similarly \cite[Eq.~(5)]{olsson1998friction}) is parametrized by a Coulomb friction constant $\bar f_c>0$ and by the viscous friction constant $\alpha_v>0$. The expression of $f_f$ reads
\begin{equation}
\label{eq:frictionforce}
f_f(f_r,v) \! := \!\! \begin{cases}
\bar f_c \sign(v) + \alpha_v v, \!\!\!\!&  \text{if } v \neq 0\\
f_r, & \text{if }  v = 0,\, |f_r| < \bar f_c\\
\bar f_c \sign(f_r), & \text{if }  v = 0,\, |f_r| \ge \bar f_c\\
\end{cases}
\end{equation}
where $f_r$ is the resultant tangential force. The mass is actuated by the PID control $u_\text{PID}$
\begin{equation}
\label{eq:uPID}
\begin{aligned}
u_{\text{PID}}(t)  & := - \bar k_p s(t)-\bar k_i \int_0^t{s(\tau)d\tau}-\bar k_d \frac{ds(t)}{dt} \\
        & = - \bar k_p s(t) - \bar k_i e_i(t)  -\bar k_d v(t),
\end{aligned}
\end{equation}
where $e_i$ is defined to be the integral of the position error and is the state of the controller, satisfying $\dot e_i = s$ and $e_i(0)=0$.

Using Newton's law, we write the mechanical dynamics $\dot s = v$ and $m \dot v = u_\text{PID} - f_f(u_\text{PID},v)$. The convenient definitions $u := \frac{u_\text{PID} - \alpha_v v}{m}$, $(k_p,k_v,k_i) := (\frac{\bar k_p}{m}, \frac{\bar k_d+\alpha_v}{m}, \frac{\bar k_i}{m})$ and $f_c:=\frac{\bar f_c}{m}$ yield then
\begin{subequations}
\label{eq:basic}
\begin{align}
\label{eq:pdot}
&\dot e_i = s \\
&\dot s = v \\
\label{eq:vDot}
&\dot v  = 
\begin{cases}
u- f_c &  \text{ if } v> 0 \text{ or } (v=0,u\ge f_c)\\
0 &  \text{ if } (v=0,|u| < f_c)\\
u+ f_c &  \text{ if } v< 0 \text{ or } (v=0,u\le -f_c)\\
\end{cases}\\
& u = -k_p s - k_v v -k_i e_i, \label{eq:u}
\end{align}
\end{subequations}
where we used that $u_\text{PID}= m\,u$ for $v=0$.

Model \eqref{eq:basic} arises from a relatively intuitive description of the mechanical principles behind the model of Figure~\ref{fig:MaPIDscheme}.
Its discontinuous right hand side makes it hard to prove existence of solutions for any initial conditions, even though such a property can be shown to hold on a case-by-case basis. Moreover, it seems to be hard to use dynamics \eqref{eq:basic} for establishing some stability properties and certifying that the position $s$ converges to zero.

In this paper we use the monotone set-valued friction law \cite[Eq.~5.36]{leine2007stability} for which existence of solutions is structurally guaranteed. Defining the overall state  $z := (e_i,s,v)$, this is equivalent to applying the Filippov \cite{filippov2013differential} or Krasovskii regularization to the discontinuous dynamics~\eqref{eq:basic} and obtain
\begin{subequations}
\label{eq:xPhDot}
\begin{equation}
\label{eq:Fdynamics}
\dot{z}
\in
\bigmat{s\\ v\\ - k_i e_i - k_p s   -k_v v} 
 - f_c \bigmat{0 \\ 0 \\ 1}\SGN(v) 
\end{equation}
where the function $\SGN$ is a set-valued map defined as
\begin{equation}
\label{eq:frictionforceSetValued}
\SGN(v):=\begin{cases} \sign(v), & \text{ if }v \neq 0\\ [-1,1], & \text{ if } v =0.\end{cases}
\end{equation}
\end{subequations}

Note that model~\eqref{eq:xPhDot} recognizes that the Coulomb friction can be selected as any force in the set $[-\bar f_c,\bar f_c]$ when $v$ is zero and has magnitude $\bar f_c$ and direction opposite to $v$ whenever $v \neq 0$. One may wonder whether any artificial solution is introduced by such an enriched description of the dynamics. (For~\eqref{eq:basic} or \eqref{eq:xPhDot}, we consider a solution to be any locally absolutely continuous function $x$ that satisfies respectively $\dot{x}(t)=f(x(t))$ or $\dot{x}(t) \in F(x(t))$ for almost all $t$ in its domain.) The following result establishes uniqueness of the solutions to~\eqref{eq:xPhDot}, which implies that the unique solution to~\eqref{eq:xPhDot} must necessarily be the unique solution to~\eqref{eq:basic}. Indeed, dynamics~\eqref{eq:basic} allows for only some selections of $\dot v$ compared to those allowed by~\eqref{eq:xPhDot}, so that any solution to~\eqref{eq:basic} is also a solution to~\eqref{eq:xPhDot}.

\begin{lemma}
\label{lem:uniq}
For any initial condition $z(0) \in \real^3$, system \eqref{eq:xPhDot} has a unique solution defined for all $t\ge 0$.
\end{lemma}

\begin{proof}
Existence of solutions follows from~\cite[\S7, Thm.~1]{filippov2013differential} because the mapping in~\eqref{eq:xPhDot} is outer semicontinuous and locally bounded with nonempty compact convex values (see also \cite[Prop.~6.10]{goebel2012hybrid}). Completeness of maximal solutions follows from local existence and no finite escape times, as \eqref{eq:xPhDot} can be regarded as a linear system forced by a bounded input.
To prove uniqueness, consider two solutions $z_1=(z_{1,e_i},z_{1,s},z_{1,v})$, $z_2$ both starting at $z_0$ and define $\delta(t)=(\delta_{e_i}(t),\delta_s(t),\delta_v(t)):=z_1(t)-z_2(t)$, for all $t\ge 0$. Then, $\delta(0)=0$ and, for almost all $t \ge 0$,
\begin{equation*}
\dot \delta(t) \in A_\delta \delta(t) -f_c \smallmat{0\\0\\1}\big( \SGN(z_{1,v}(t)) - \SGN(z_{1,v}(t)-\delta_v(t)) \big),
\end{equation*}
with $A_\delta := \smallmat{0 & 1 & 0\\ 0 & 0 & 1\\ -k_i & -k_p & -k_v}$, whose maximum singular value is $\lambda_\delta$. Therefore we can write for almost all $t$
\begin{flalign*}
&\tfrac{d}{dt} \tfrac{|\delta(t)|^2}{2} = \delta(t)^T \dot \delta(t) \le  \lambda_\delta |\delta(t)|^2 + M(t)\\
& M(t):= \!\!\!\!\!\!\!\!\!\!\!\!\!\!  \max_{\substack{f_1\in f_c\SGN(z_{1,v}(t))\\f_2\in f_c \SGN(z_{1,v}(t)-\delta_v(t))}} \!\!\!\!\!\!\!\!\!\!\! \delta_v(t) (f_2 -f_1).
\end{flalign*}

Whether $z_{1,v}(t)$ and $z_{1,v}(t)-\delta_v(t)$ are positive, zero or negative, by trivial inspection of all the cases it can be shown that $M(t)\le 0$ for all $t\ge 0$. Therefore,
\begin{equation*}
\tfrac{d}{dt}\tfrac{|\delta(t)|^2}{2} \le \lambda_\delta |\delta(t)|^2 \text{ for almost all } t\ge 0,
\end{equation*}
and from standard comparison theorems $\delta(0)=0$ implies $\delta(t)=0$ for all $t\ge 0$, that is, $z_1(t)=z_2(t)$ for all $t\ge0$.
\end{proof}

\subsection{Main result}

The advantage in the use of the compact dynamics \eqref{eq:xPhDot} is that we may adopt Lyapunov tools to study the asymptotic stability properties of the rest position under the following standard assumption
(see, e.g., \cite{armstrong1996pid}).
\begin{assumption}
\label{ass:PIDpars}
The parameters in~\eqref{eq:u} are such that
\begin{equation*}
k_i>0,\,k_p>0,\, k_v k_p > k_i.
\end{equation*}
\end{assumption}

According to the Routh stability test, Assumption~\ref{ass:PIDpars} holds if and only if the origin of the dynamics in~\eqref{eq:xPhDot} with $f_{c}=0$ is globally exponentially stable.

Under Assumption~\ref{ass:PIDpars}, one readily sees that all possible equilibria of dynamics \eqref{eq:xPhDot} correspond to $(e_i,s,v)=(\bar e_i,0,0)$ with $|\bar e_i|\le \frac{f_c}{k_i}$, that is, whenever the mass is at rest at zero position and the size of the integral error $e_i$ is bounded by the specific threshold $\frac{f_c}{k_i}$. Any of these points is an equilibrium for~\eqref{eq:xPhDot} because in~\eqref{eq:xPhDot} a value can be selected from $f_c \SGN(0)$ such that the (unique) solution maintains $\dot z$ identically zero. Note that here we consider the problem of tracking a position setpoint $s^o=0$, but this is trivially generalized to piecewise constant setpoints, thanks to the global nature of our results and a trivial change of coordinates.
Denote then the set of these equilibria
as
\begin{equation}
\label{eq:Aph}
\A := \left\{(e_i,s,v) : \ s=0, \ v=0, e_i\in \bigg[-\frac{f_c}{k_i},\frac{f_c}{k_i}\bigg] \right\}.
\end{equation}

Given the attractor $\A$, we prove its global attractivity in Section~\ref{sec:attr} and its stability in Section~\ref{sec:stab}. For the attractivity, we use a suitable discontinuous Lyapunov-like function and a nonsmooth version of LaSalle's invariance principle~\cite[\S4.2]{khalilNonlinear}. These tools are applicable to our scenario because of the desirable structural properties of the regularization in~\eqref{eq:xPhDot}. 

\begin{proposition}
\label{prop:GA+S}
Under Assumption~\ref{ass:PIDpars}, the attractor $\A$ in~\eqref{eq:Aph} is 1) globally attractive and 2) Lyapunov stable for dynamics~\eqref{eq:xPhDot}.
\end{proposition}

Note that no smaller set could be proven to be globally attractive because $\A$ is a union of equilibria.

Our main result in Theorem~\ref{thm:GAS} establishes relevant robust stability properties of~$\A$ for~\eqref{eq:xPhDot} in terms of the behavior of solutions to the perturbed dynamics\footnote{$\tilde F$, $\ball$ and $\overline{\text{co}}$ denote respectively the (set-valued) right-hand side of \eqref{eq:Fdynamics}, the closed unit ball and the closed convex hull of a set.}
\begin{equation}
\label{eq:xPhDotPert}
\dot z \in \overline{\text{co}} \tilde F(z+\rho(z)\ball)+ \rho(z)\ball,
\end{equation}
where $\rho:\real^3 \rightarrow \real_{\ge 0}$ is a suitable non-negative perturbation function satisfying $z\notin \A \Rightarrow  \rho(z)>0$. In particular, by using the equivalences in \cite[Chap.~7]{goebel2012hybrid} (which apply because system~\eqref{eq:xPhDot} is well-posed from the regularity of $\tilde F$ \cite[Thm~6.30]{goebel2012hybrid} and $\A$ is compact), Proposition~\ref{prop:GA+S} (see~\cite[Thm.~7.21]{goebel2012hybrid}) implies robust global $\mathcal{KL}$ asymptotic stability of $\A$ for~\eqref{eq:xPhDot}, namely the existence of $\beta_0\in\mathcal{KL}$ such that all solutions to~\eqref{eq:xPhDotPert} satisfy $|z(t)|_\A\le \beta_0(|z(0)|_\A,t)$ for all $t\ge 0$, in addition to robust global uniform asymptotic stability of~$\A$ for~\eqref{eq:xPhDot}, namely the property that $\A$ is uniformly globally stable and attractive for~\eqref{eq:xPhDotPert} (see \cite[Def.~3.6]{goebel2012hybrid}).
\begin{theorem}
\label{thm:GAS}
Under Assumption~\ref{ass:PIDpars}, the attractor $\A$ in~\eqref{eq:Aph} is robustly globally $\mathcal{KL}$ asymptotically stable for dynamics~\eqref{eq:xPhDot} and robustly uniformly globally asymptotically stable.
\end{theorem}

An interesting consequence of the robustness result established in Theorem~\ref{thm:GAS} is the semiglobal practical robust asymptotic stability of attractor $\A$ established in \cite[Thm 7.12 and Lemma 7.20]{goebel2012hybrid}. A specific perturbation of interest arises when selecting a constant scalar $\rho_v \in \real$ and perturbing the friction effect as follows:
\begin{align}
\label{eq:perturbed}
\dot{z}
&\in
\smallmat{s\\ v\\ - k_i e_i - k_p s   -k_v v} 
 - f_c \smallmat{0 \\ 0 \\ 1}{\SGN}_{\rho_v}(v) \\
 \nonumber
{\SGN}_{\rho_v}(v)& :=\begin{cases} 
[\sign(v)-|\rho_v|, \sign(v)+|\rho_v|], &\!\! \text{if }|v|  > |\rho_v| \\
[-1-|\rho_v|,1+|\rho_v|], &\!\! \text{if }|v| \leq |\rho_v| .\end{cases}
\end{align}

The general definition of the inflation of a set-valued mapping for $\SGN$ in~\eqref{eq:frictionforceSetValued} is $\SGN(v+|\rho_v|\ball)+|\rho_v|\ball$. This inflation coincides with $\SGN_{\rho_v}$ in~\eqref{eq:perturbed}, and in the special case $\rho_v=0$, $\SGN_{0}$ clearly coincides with $\SGN$.

\begin{corollary}
\label{cor:stribeck}
Under Assumption~\ref{ass:PIDpars}, the attractor $\A$ in~\eqref{eq:Aph} is globally input-to-state stable for dynamics~\eqref{eq:perturbed} from input $\rho_v$.
\end{corollary}

\smallskip

\begin{proof}
The solutions to~\eqref{eq:perturbed}  are a subset of the solutions to
$\dot{z} = A_\delta z - f_c \smallmat{0\\0\\1} m$,
where $A_\delta$ was trivially defined in the proof of Lemma~\ref{lem:uniq} and is Hurwitz from Assumption~\ref{ass:PIDpars},
and $m$ is a locally integrable signal that for the constant scalar $\rho_v$ and for all $t$ satisfies $m(t) \le 1 + |\rho_v|$ because, for all $t$, ${\SGN}_{\rho_v}(v(t)) \le  1 + |\rho_v|$.
From BIBO stability of exponentially stable linear systems, 
there exist positive $c$ and $\lambda$ such that all solutions satisfy
\begin{equation}
\label{eq:norm2bound}
|z(t)| 
\le c e^{-\lambda t} |z(0)| + c (1+|\rho_v|).
\end{equation}
From the two distances
$|z|_\A^2  :=s^2 + v^2 + \big(\mathrm{dz}_{{f_c}/{k_i}}(e_i)\big)^2$,
$|z|^2  :=s^2 + v^2 + e_i^2$,
we have $|z|_\A \le |z|$ and $|z|^2 \le 2 |z|_\A^2 + 2 \big( \tfrac{f_c}{k_i} \big)^2$ (by splitting into the cases $|e_i|\ge \tfrac{f_c}{k_i}$ and $|e_i| < \tfrac{f_c}{k_i}$), which implies $|z| \le \sqrt{2} \big( |z|_\A + \tfrac{f_c}{k_i} \big)$. These relationships between the two distances and \eqref{eq:norm2bound} imply that there exist positive constants $\kappa_1$, $\kappa_2$, $\kappa_3$ such that all solutions satisfy
\begin{equation}
\label{eq:KL+N}
\begin{split}
|z(t)|_\A & \le |z(t)| \le c e^{-\lambda t} |z(0)| + c (1+|\rho_v|)\\
 & \le \kappa_1 e^{-\lambda t} |z(0)|_\A + \kappa_2  +\kappa_3 |\rho_v|,\,\forall t\ge 0.
\end{split}
\end{equation}
Using Theorem~\ref{thm:GAS} and
the semiglobal practical robustness of $\mathcal{KL}$ asymptotic stability established in \cite[Lemma 7.20]{goebel2012hybrid}, one can transform the $\delta$-$\epsilon$ argument of \cite[Lemma 7.20]{goebel2012hybrid} into a class $\mathcal{K}$ function $\gamma_\ell$ by following similar steps to \cite[Lemma 4.5]{khalilNonlinear}. Moreover, using a similar approach to \cite[Thm. 2]{sontag1990further} relating the size of the initial condition and of the input, we obtain the following:
\begin{multline}
\label{eq:localISS}
|z(0)|_{\mathcal A} \leq \tfrac{1}{\delta_\ell}, |\rho_v| \leq \delta_\ell \Rightarrow\\
|z(t)|_{\mathcal A} \leq \beta_\ell (|z(0)|_{\mathcal A} ,t )+ 
\gamma_\ell(|\rho_v|), \; \forall t\geq 0,
\end{multline}
for some suitable class $\mathcal{KL}$ and class $\mathcal{K}$ functions $\beta_\ell$ and $\gamma_\ell$, and for a small enough scalar $\delta_\ell>0$.
Without loss of generality, consider now using in \eqref{eq:localISS} a small enough $\delta_\ell$ such that $(2\delta_\ell)^{-1} \geq \kappa_2 + \kappa_3 \delta_\ell$. 
Introduce function $T^\star \colon \real_{\geq 0} \to \real_{\geq 0}$ with $T^\star (s):= \max \{0, \lambda^{-1} \log (2 \delta_\ell \kappa_1 s)\}$, which satisfies:
\begin{equation}
\label{eq:Tgood}
\kappa_1\, e^{-\lambda T^\star (s)}\,\, s + \kappa_2  +\kappa_3 \delta_\ell \leq 
\delta_\ell^{-1}, \quad \forall s\geq 0.
\end{equation}

\begin{figure}[!t]
\centering
\includegraphics[width=2.5in]{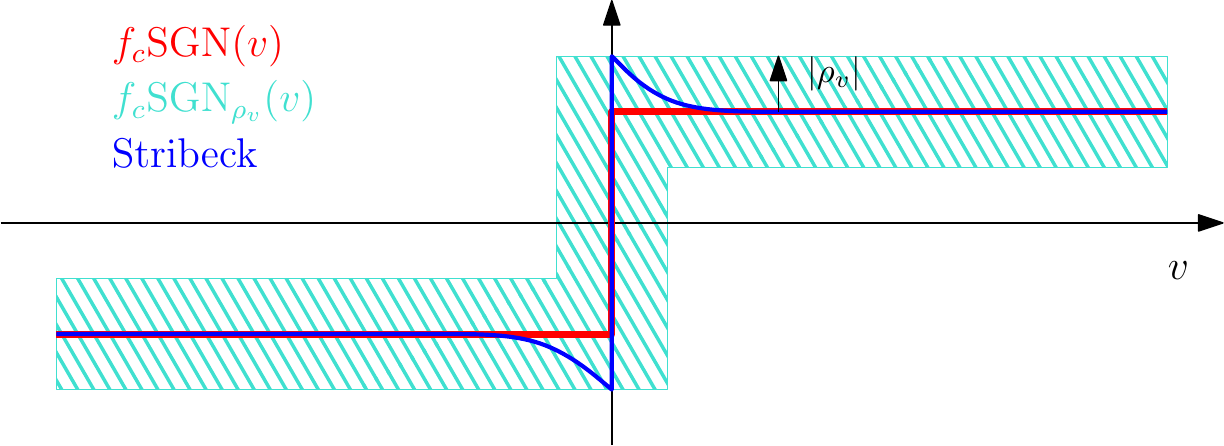}
\caption{Stribeck effect is included in the perturbation~\eqref{eq:perturbed}.}
\label{fig:stribeck}
\end{figure}

Finally, we conclude the proof by establishing the following (global) ISS bound from $\rho_v$:
\begin{equation}
\label{eq:ISS}
 |z(t)|_\A \le \beta(|z(0)|_\A,t) + \gamma(|\rho_v|), \forall z(0), \forall \rho_v, \forall t\ge 0,
\end{equation}
where functions $\beta$ and $\gamma$ of class $\mathcal{KL}$ and class $\K$, respectively, are built starting from the following inequalities:%
\begin{subequations}
\label{eq:betaGamma}
\begin{empheq}[left={\!\!\!\!\beta(s,t) \ge\! \empheqlbrace}]{align}
\kappa_1 e^{-\lambda t} s + \kappa_2  +\kappa_3 \delta_\ell, & \text{ if }  s\geq \tfrac{1}{\delta_\ell},\, t\leq T^\star(s)  \label{eq:betaTop} \\
\mathfrak{b}(s,t), \hspace{2.2cm} & \text{ otherwise} \label{eq:betaBottom}
\end{empheq}\vspace{-0.25cm}
\begin{equation}
\label{eq:betaBottomAlias}
\!\,\mathfrak{b}(s,t):=\max\left\{ \beta_\ell \left(s,\max\{0,t-T^\star(s)\}\right),  \kappa_1 e^{-\lambda t} s \right\} 
\end{equation}\vspace{-0.5cm}
\begin{empheq}[left={\hspace*{-4.75cm} \gamma(s) \ge \!\empheqlbrace}]{align}
\kappa_2 + \kappa_3 s, & \text{ if } s \ge  \delta_\ell \label{eq:gammaTop} \\
\gamma_\ell (s), \hspace{0.5cm} & \text{ if }s \le \delta_\ell. \label{eq:gammaBottom}
\end{empheq}
\end{subequations}
The effectiveness of selections \eqref{eq:betaGamma} for establishing the ISS bound \eqref{eq:ISS} can be verified case by case.\\
{\bf Case 1} ($|\rho_v| \geq \delta_\ell$): use \eqref{eq:KL+N}, \eqref{eq:gammaTop}, and bound $\kappa_1 e^{-\lambda t} s$ in \eqref{eq:betaBottomAlias}-\eqref{eq:betaBottom}.\\
{\bf Case 2} ($|\rho_v| \leq \delta_\ell$ and
 $|z(0)|_\A \leq \delta_\ell^{-1}$):  use \eqref{eq:localISS}, \eqref{eq:gammaBottom}, and  bound $\beta_\ell \left(s,\max\{0,t-T^\star(s)\}\right)$ in \eqref{eq:betaBottomAlias}-\eqref{eq:betaBottom}.\\
{\bf Case 3} ($|\rho_v| \leq \delta_\ell$ and
 $|z(0)|_\A \geq \delta_\ell^{-1}$): for $t\leq T^\star(|z(0)|_\A)$  use \eqref{eq:betaTop} and nonnegativity of $\gamma$, whereas for $t\geq T^\star(|z(0)|_\A)$  use 
 $|z(T^\star(|z(0)|_\A))|_\A \leq {\delta_\ell}^{-1}$ (from \eqref{eq:KL+N} and \eqref{eq:Tgood}) and the semigroup property of solutions to fall again into case 2 above.
\end{proof}

A consequence of Corollary~\ref{cor:stribeck} is that the Stribeck effect, which is known to lead to persistent oscillations (the so-called \emph{hunting} phenomenon), produces solutions that are graceful degradations (in the ISS sense) of the asymptotically stable solutions to the unperturbed dynamics because small Stribeck deformations lead to graphs included in the graph of $f_c \SGN_{\rho_v}(v)$, as shown in Figure~\ref{fig:stribeck}.

\begin{table}[tb]
\caption{Parameters and eigenvalues for simulations in Section~\ref{sec:sims}.} 
\label{tab:PIDpar} 
\centering 
\begin{tabular}{m{0.25cm}m{0.35cm}m{0.35cm}m{0.35cm}m{2.7cm}}
\toprule
Case & $k_v$ & $k_p$ & $k_i$ & Roots\\
\midrule
\textit{(a)} & 6.4 & 3       & 4              & $-6.01$, $-0.19\pm i 0.79$\\
\textit{(b)} & 1.5 & 0.66  & 0.08     & $-0.8$, $-0.5$, $-0.2$\\
\bottomrule
\end{tabular}
\end{table}

\section{Illustration by simulation}
\label{sec:sims}

\begin{figure*}[tb]
\centering
{\includegraphics[width=0.48\textwidth]{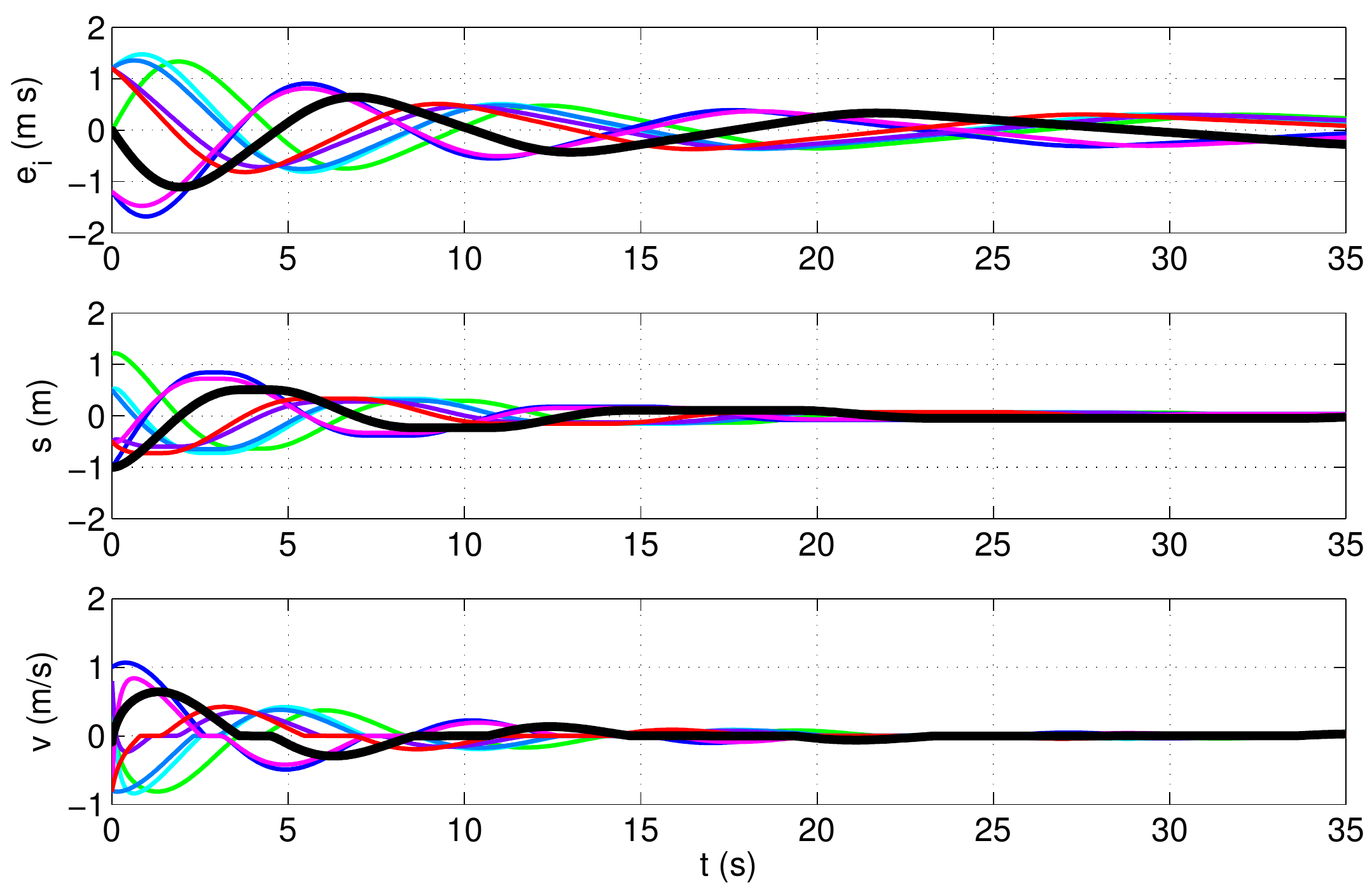}} \,\,
{\includegraphics[width=0.48\textwidth]{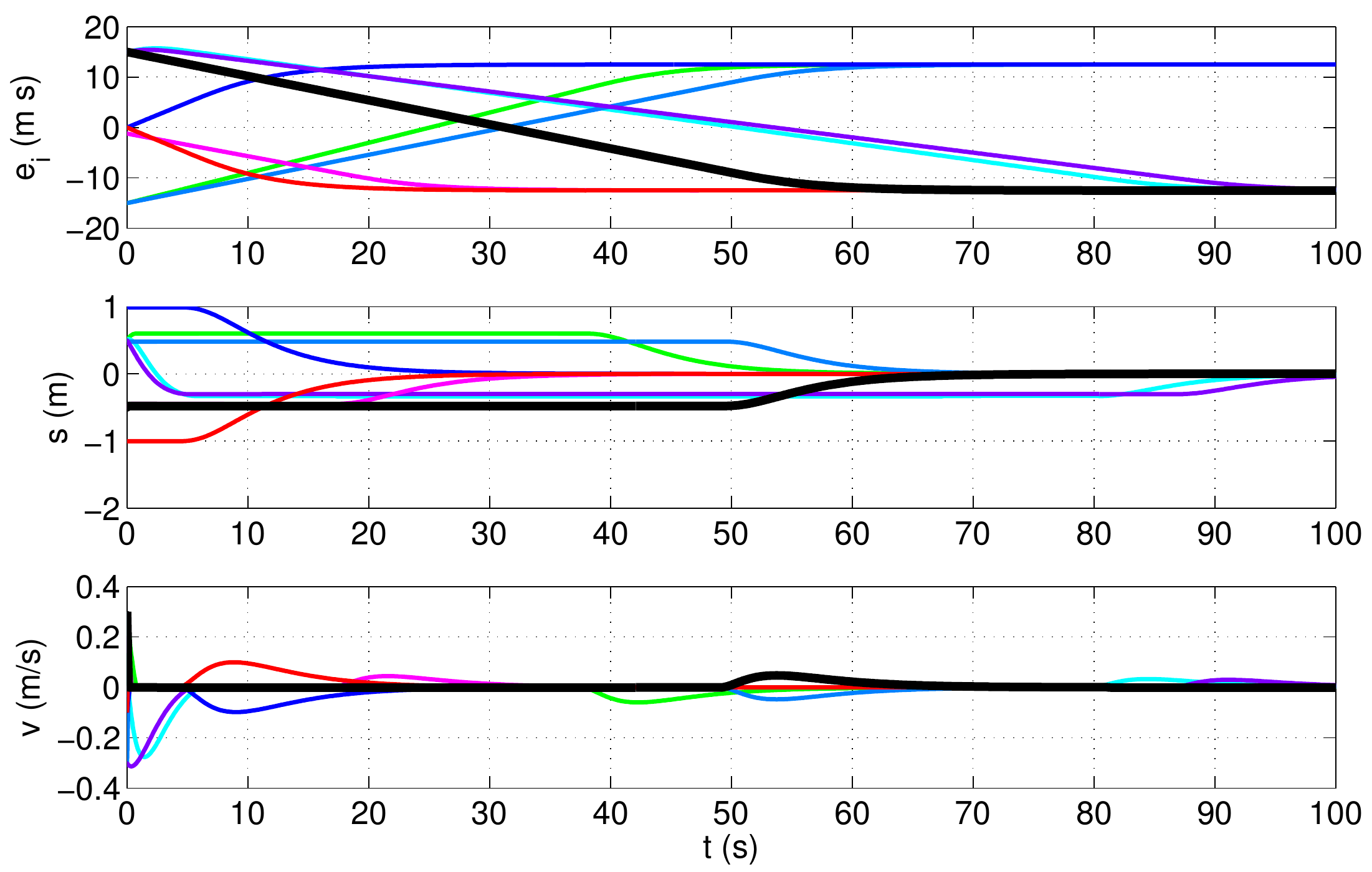}} \\
{\includegraphics[width=0.48\textwidth]{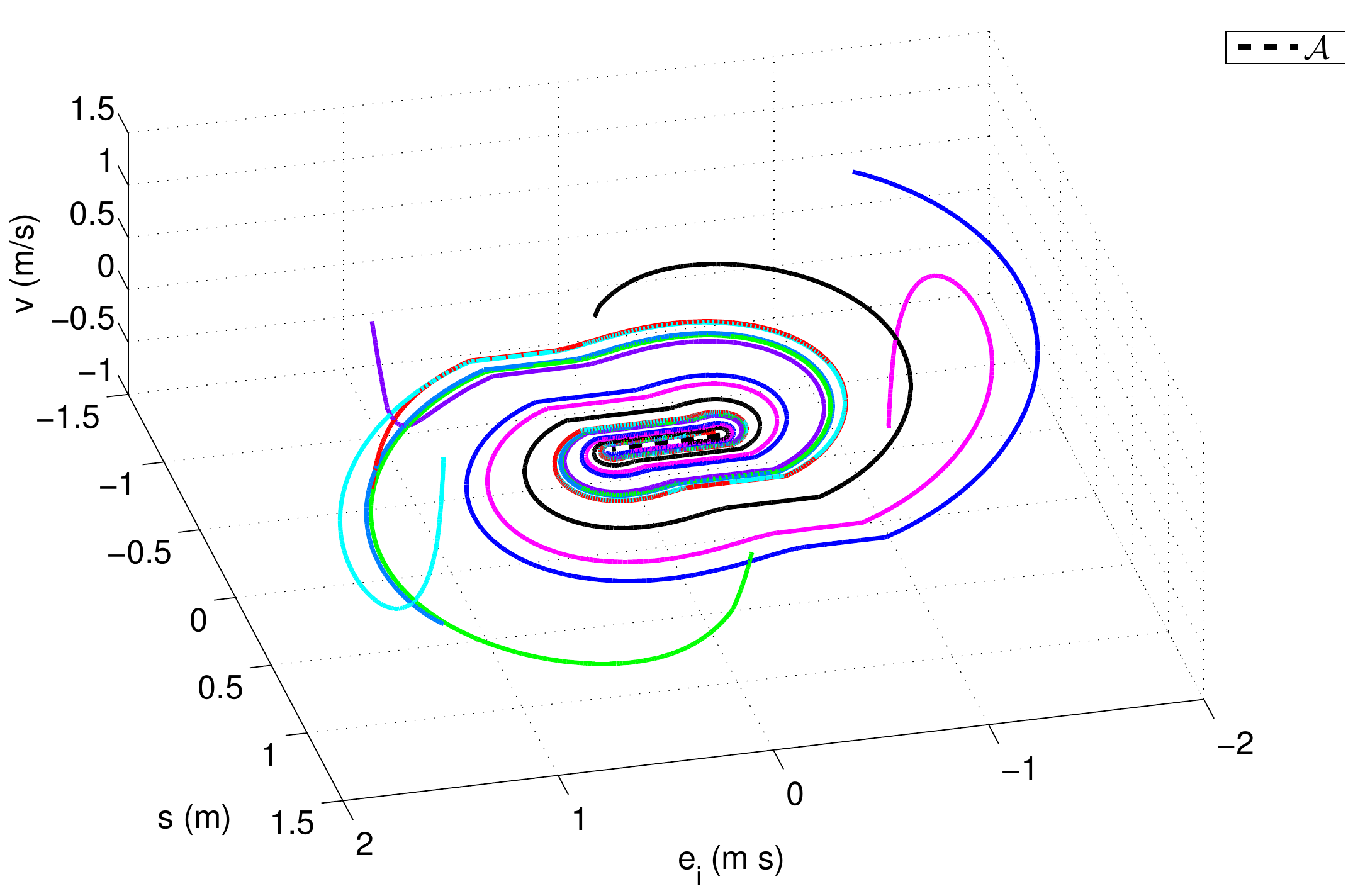}}\,\,
{\includegraphics[width=0.48\textwidth]{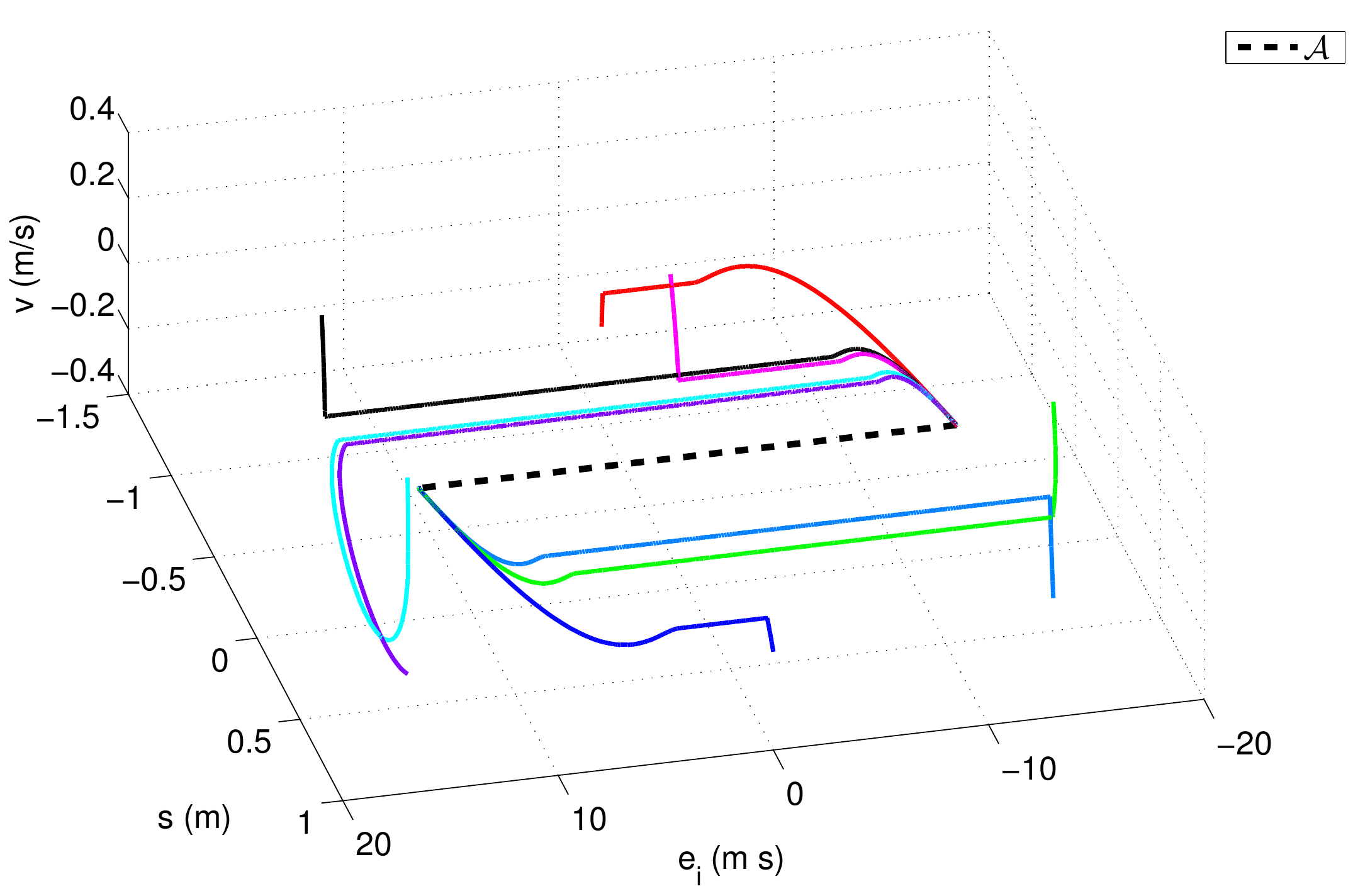}}\\
{\includegraphics[width=0.48\textwidth]{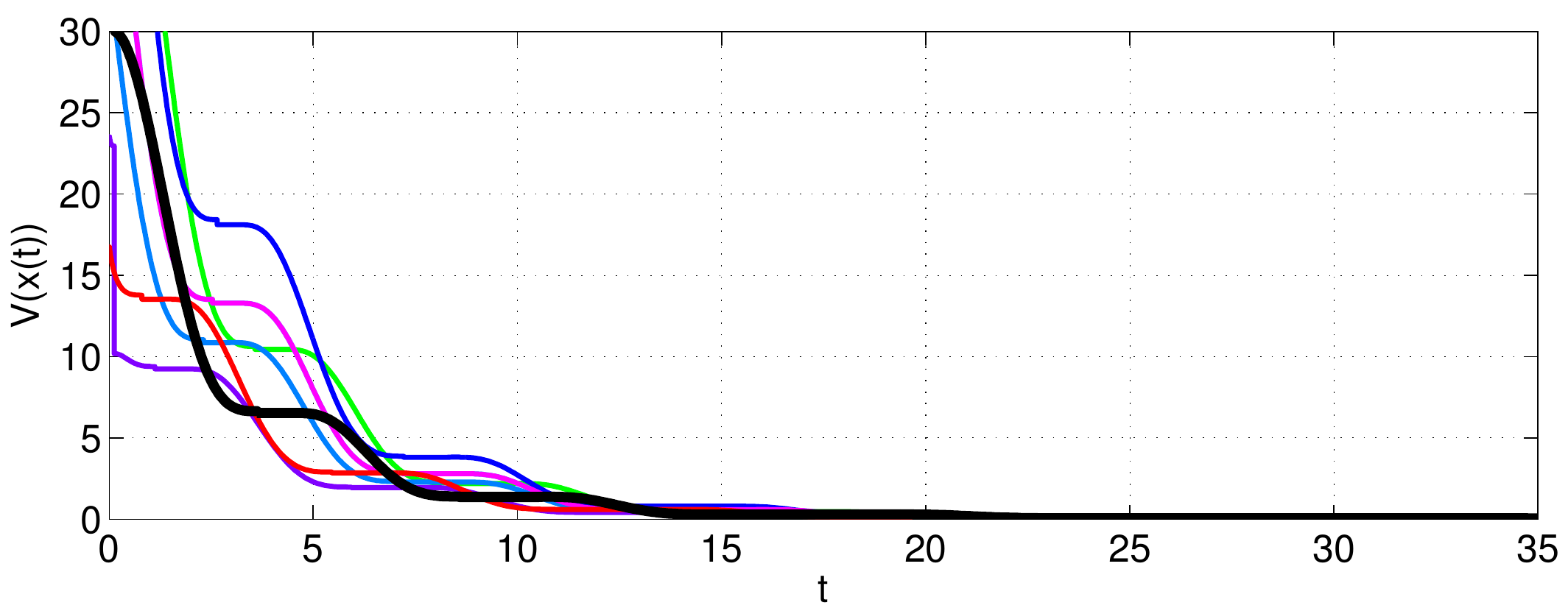}} \,\,
{\includegraphics[width=0.48\textwidth]{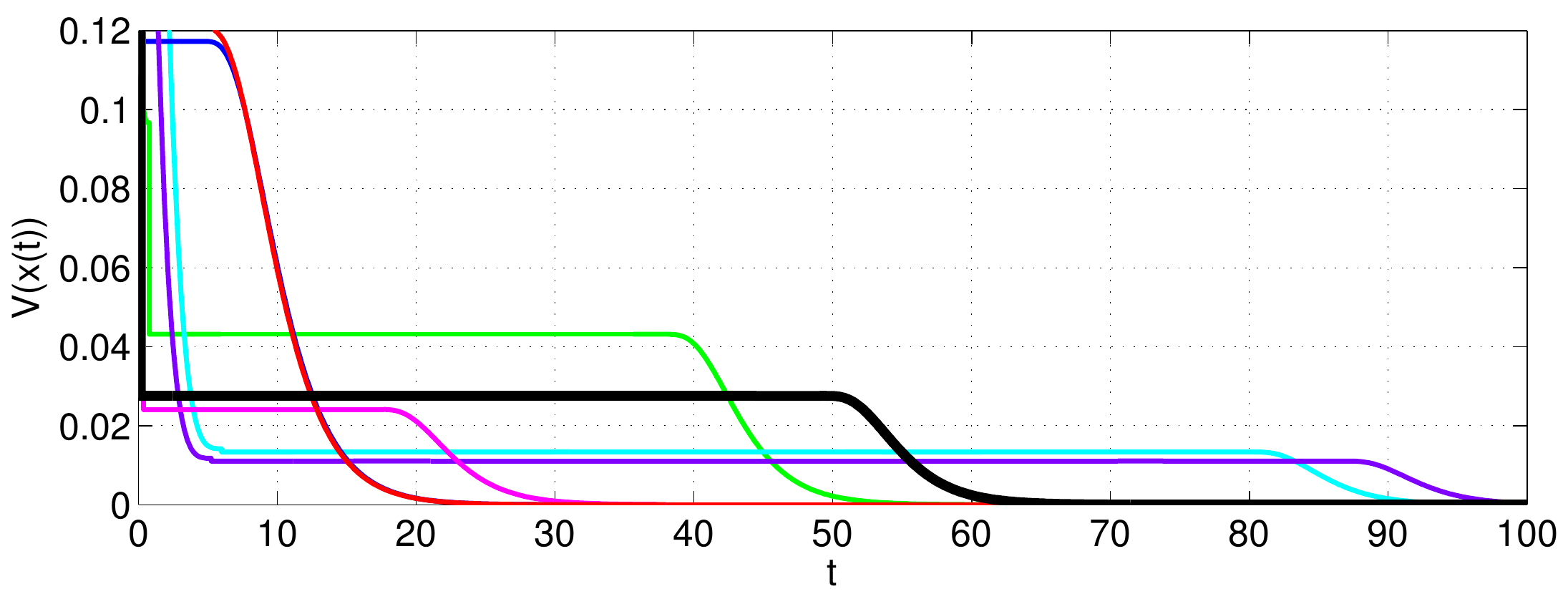}}
\caption{Top: solutions to~\eqref{eq:xPhDot} for different initial conditions. Center: phase portraits for~\eqref{eq:xPhDot} for the same solutions. Bottom: Lyapunov-like function $V$ in~\eqref{eq:lyapLike} evaluated along the same solutions. All the figures to the left (resp., right) refer to the PID parameters \textit{(a)} (resp., \textit{(b)}) in Table~\ref{tab:PIDpar}.}
\label{fig:sims}
\end{figure*}

Before we prove our main result, we obtain solutions by simulation in order to illustrate the typical behaviour of~\eqref{eq:xPhDot} and the convergence to the attractor. Simulations capture for each initial condition the unique solution to~\eqref{eq:xPhDot} because of Lemma~\ref{lem:uniq}.

When $f_c=0$, \eqref{eq:xPhDot} reduces to a linear system with characteristic polynomial $s^3+k_v s^2+k_p s + k_i =0$. From Assumption~\ref{ass:PIDpars}, its roots have negative real part, but we do not assume anything about their specific locations in the complex plane. In~\cite[Lemmas~L1 and L2]{armstrong1996pid}, each of these possible locations needed examining as part of a proof relying on the algebraic expression of solutions. We will not need this in our subsequent proof, but we base our simulations on two representative cases, complex conjugate and three distinct real roots. The corresponding parameters $k_v$, $k_p$, $k_i$ and roots are listed in Table~\ref{tab:PIDpar} and the value $f_c=1$~m/s$^2$ is common to all simulations.

First, we present the solutions to~\eqref{eq:xPhDot} for different sets of initial conditions $(\bar e_i, \bar s,\bar v)$ for the cases \textit{(a)} and \textit{(b)} in Table~\ref{tab:PIDpar}, respectively in the left and right top plots of Figure~\ref{fig:sims}. We use a heavier black line to point out a specific solution and the two different phases that are visible on it. In particular, there are time intervals when the mass is in motion (called slip phases in the friction literature) and others when the mass is at rest (called stick phases), which can be detected easily from the velocity being zero on a nonzero time interval. Whenever the mass is in a slip phase, the PID control acts in the direction of getting the mass closer to the position setpoint at zero. During a stick phase starting at $t_i$, the mass is at rest ($v=0$) and only the error integral builds up linearly in time as $e_i(t)=e_i(t_i)+s(t_i) (t-t_i)$ until the control action $u$ is such to overcome the Coulomb friction, that is, $|u|=|-k_i e_i -k_p s|=f_c$. So, the closer the mass is to zero position (smaller $s(t_i)$), the longer it takes the error to build up. This explains the ramps in the top of Figure~\ref{fig:sims}, and their decreasing slope and increasing duration. Moreover, we notice that position and velocity converge to zero, but the error integral does not in general. For the case of complex conjugate roots (top, left), the error integral continues to oscillate and these oscillations enter asymptotically the set $[-f_c/k_i,f_c/k_i]$ as the position approaches zero. For the case of distinct real roots (top, right), after a stick phase the position and the velocity converge to zero \emph{only} asymptotically, so that $v$ remains always nonzero and in~\eqref{eq:xPhDot} $e_i$ can only approach the equilibria $f_c/k_i$ or $-f_c/k_i$. As a consequence of the reasoning above (the smaller the position, the longer it takes to exit a stick phase), solutions converge asymptotically, but not exponentially.

Second, we present in the left and right center plots of Figure~\ref{fig:sims} a phase portrait for the same solutions in the top plots. In these figures it is evident that solutions converge to the attractor in~\eqref{eq:Aph}, as we postulate in item 1) of Proposition~\ref{prop:GA+S}. In particular, the two different manners in which solutions converge to the attractor in the cases of complex conjugate or distinct real roots are evident here as well. In the left center plot, we can also appreciate the presence of a strip in the plane $v=0$ expressed by equation $-f_c \le -k_i e_i - k_p s \le f_c$ (as in \cite[Page~7]{radcliffe1990property}), that is, the region of the state space where stick is bound to occur.

Third, we keep the same initial conditions and parameters and we anticipate the evolution along solutions of the Lyapunov-like function introduced in the next section (see~\eqref{eq:lyapLike}). In particular, this function is nonincreasing along solutions, it can be discontinuous (e.g., the left, bottom, violet curve at $t=0.123$~s), and remains constant during stick (as pointed out by the same heavier black curves).

\section{Proof of Prop.~\ref{prop:GA+S}: global attractivity}
\label{sec:attr}

\subsection{Coordinate change and discontinuous LaSalle function}
\label{sec:Lyap}

For the following analysis we adopt a specific change of coordinates for~\eqref{eq:xPhDot}, that is,
\begin{equation}
\label{eq:CoC}
\begin{aligned}
\sigma & := - k_i s\\
\phi & := - k_i e_i -k_p s, \\
v & := v.
\end{aligned}
\end{equation}
The change of coordinates is nonsingular thanks to Assumption~\ref{ass:PIDpars} ($k_i$, $k_p$ strictly positive) and it rewrites (\ref{eq:xPhDot}) as
\begin{equation}
\label{eq:xDot}
\begin{aligned}
\dot{x}:=
\begin{bmatrix}
\dot{\sigma}\\
\dot{\phi}\\
\dot{v}
\end{bmatrix}
& \in
\begin{bmatrix}
- k_i v\\
\sigma - k_p v\\
\phi - k_v v -f_c \SGN (v)
\end{bmatrix}
\\
& =
\underbrace{
\begin{bmatrix}
0 & 0 & -k_i\\
1 & 0 & -k_p\\
0 & 1 & -k_v
\end{bmatrix}
}_{=:A}
\begin{bmatrix}
\sigma\\
\phi\\
v
\end{bmatrix}
-
\underbrace{
\begin{bmatrix}
0\\
0\\
f_c 
\end{bmatrix}}_{=:b}\SGN(v)\\
& =A x - b \SGN(v) =: F(x).
\end{aligned}
\end{equation}

In the new coordinates $x$, the attractor ${\mathcal A}$ in \eqref{eq:A} can be expressed as
\begin{equation}
\label{eq:A}
\A= \{ 	(\sigma,\phi,v) \colon |\phi| \le f_c, \sigma = 0, v=0 \}.
\end{equation}
Among other things, the simple expression in~\eqref{eq:A} allows writing
explicitly the distance of a point $x$ from $\A$ as
\begin{equation}
\label{eq:distFromA}
|x|_\A ^2 :=\big(\inf_{y\in\A} |x-y| \big)^2=\sigma^2+v^2+\dz(\phi)^2.
\end{equation}
where $\dz(\phi):=\phi-\text{sat}_{f_c}(\phi)$ is the symmetric scalar deadzone function returning zero when $\phi\in[-f_c,f_c]$.
Indeed, \eqref{eq:distFromA} follows from separating the cases $\phi<-f_c$, $|\phi|\le f_c$, $\phi>f_c$ and applying the definition given in~\eqref{eq:distFromA}.

Based on the above observation, it is rather intuitive to introduce the following discontinuous Lyapunov-like function
\begin{subequations}
\label{eq:lyapLike}
\begin{equation}
\label{eq:lyapLikeV}
\begin{aligned}
V(x) := & \begin{bmatrix}
\sigma\\
v
\end{bmatrix}^T
\begin{bmatrix}
\frac{k_v }{k_i} & -1\\
-1 & k_p 
\end{bmatrix}
\begin{bmatrix}
\sigma\\
v
\end{bmatrix}+\!\!
\min\limits_{f\in f_c \SGN(v)}  |\phi-f|^2\\ 
= &
\min\limits_{f\in f_c \SGN(v)} \smallmat{\sigma \\ \phi-f \\ v}^T P \smallmat{\sigma \\ \phi-f \\ v}
\end{aligned}
\end{equation}
where the matrix $P$ is given by
\begin{equation}
\label{eq:lyapLikeP}
P:=
\begin{bmatrix}
\frac{k_v }{k_i} & 0 & -1\\
0 & 1 &0\\
-1 & 0 & k_p 
\end{bmatrix}.
\end{equation}
\end{subequations}
Note that for $v\neq 0$ the minimization in~\eqref{eq:lyapLikeV} becomes trivial because $f$ can take only the value $f_c \sign(v)$.
It is emphasized that function $V$ is discontinuous. For example, if we evaluate $V$ along the sequence of points $(\sigma_i,\phi_i,v_i)=(0,0,\varepsilon_i)$ for $\varepsilon_i \in (0,1)$ converging to zero, $V$ converges to $f_c^2$, even though its value at zero is zero. Nevertheless, function $V$ enjoys a number of useful properties established in the next lemma whose proof is given in Section~\ref{sec:proofLemLyapLike} below.

\begin{lemma}
\label{lem:lyapLike}
The Lyapunov-like function in~\eqref{eq:lyapLike} is lower semicontinuous (lsc) and enjoys the following properties:
\begin{enumerate}
\item \label{lem:lyapLike:item:lowerbound} $V(x)=0$ for all $x\in\A$ and there exists $c_1>0$ such that $c_1 |x|_\A^2 \le V(x)$ for all $x\in\real^3$,
\item \label{lem:lyapLike:item:decrease} there exists $c>0$ such that each solution $x=(\sigma,\phi,v)$ to~\eqref{eq:xDot} satisfies for all $t_2\ge t_1\ge 0$
\begin{equation}
\label{eq:Vdecrease}
V(x(t_2))-V(x(t_1)) \le - c \int_{t_1}^{t_2} v(t)^2 dt.
\end{equation}
\end{enumerate}
\end{lemma}

\begin{remark}
In~\cite{armstrong1996pid} it is proven that if a solution is in a \emph{slip phase} in the nonempty time interval $(t_i,t_{i+1})$ (namely, for all $t\in(t_i,t_{i+1})$, $v(t)\neq 0$) and the slip phase is preceded and followed by a \emph{stick phase} (namely, there exist $\delta>0$ such that, for all $t\in[t_i-\delta,t_i]\cup[t_{i+1},t_{i+1}+\delta]$, $v(t)=0$ and $|\phi(t) |\le f_c$), then 
\begin{equation}
\label{eq:armstrongRes}
|\sigma(t_{i+1})|<|\sigma(t_i)|.
\end{equation}
Instead of using the explicit form of solutions as~\cite[Lemma~L2]{armstrong1996pid} depending on the nature of the eigenvalues of $A$, \eqref{eq:armstrongRes} is easily concluded from~\eqref{eq:Vdecrease}, the definition~\eqref{eq:lyapLikeV}, and $|\phi(t_i)|\le f_c$, $|\phi(t_{i+1})|\le f_c$.
\end{remark}

\subsection{Proof of item 1) of Proposition~\ref{prop:GA+S} (global attractivity)}

We can now prove the first item of Proposition~\ref{prop:GA+S} based on Lemma~\ref{lem:lyapLike} and a generalized version for differential inclusions of the invariance principle \cite[\S4.2]{khalilNonlinear}. The following fact comes indeed from specializing the result in~\cite[Thm.~2.10]{ryan1998integral} to our case, where the differential inclusion \eqref{eq:xPhDot} has actually unique solutions defined for all nonnegative times (as established in Lemma~\ref{lem:uniq}). We also select $G=\real^3$, $U=\real^3$ in~the original result of \cite{ryan1998integral}.

\begin{fact}%
\label{fact:invariance}%
\cite{ryan1998integral} Let $\ell: \real^3 \rightarrow \real_{\ge 0}$ be lower semicontinuous and such that $\ell(x)\ge 0$, for all $x\in\real^3$. If $x$ is a complete and bounded solution to~\eqref{eq:xPhDot} satisfying $\int_0^{+\infty} \ell(x(t)) dt<+\infty$,
then $x$ converges to the largest forward invariant subset $\M$ of $\Sigma := \{x\in \real^3 \colon  \ell(x) = 0\}$.
\end{fact}

\medskip

\begin{proofof}{\em item 1) of Proposition~\ref{prop:GA+S} (global attractivity of $\A$).}
The proof exploits Fact~\ref{fact:invariance}, where we take $\ell(x)=v^2$. From Lemma~\ref{lem:lyapLike}, $V(x(t))\le V(x(0))$ (item~\ref{lem:lyapLike:item:decrease}) and $c_1|x(t)|_\A^2 \le V(x(t))$ (item~\ref{lem:lyapLike:item:lowerbound}), so that $c_1|x(t)|_\A^2 \le V(x(0))$ and consequently all solutions to~\eqref{eq:xDot} are bounded (their completeness is established in Lemma~\ref{lem:uniq}). Apply~\eqref{eq:Vdecrease} from $0$ to $t$, and obtain $c \int_0^t v^2(\tau) d\tau \le V(x(0))-V(x(t)) \le V(x(0))$ because $V(x(t)) \ge 0$ from Lemma~\ref{lem:lyapLike}, item~\ref{lem:lyapLike:item:lowerbound}. Then we have $\int_0^t v^2(\tau) d\tau \le \frac{V(x(0))}{c}$, and if $t\to +\infty$ we get the required boundedness of the integral of $\ell(x(\cdot))$. Then Fact~\ref{fact:invariance} guarantees that the solution converges to the largest {forward} invariant subset $\M$ of $\Sigma=\{x \colon v=0\}$. We claim that such a subset is $\A$. Indeed, $\M \subset \Sigma$ implies $v=0$ in $\M$. Moreover, $\sigma=0$ in $\M$ because each solution starting from $v=0$ and $\sigma\neq 0$ causes a ramp of $\phi$ that eventually reaches $|\phi| > f_c$ and drives $v$ away from zero (therefore out of $\Sigma$). Finally, in $\M$ we must have $|\phi| \le f_c$ otherwise $v$ would become nonzero again.
Therefore the largest {forward} invariant set $\M$ in $\Sigma$ is the attractor $\A$.
\end{proofof}

\subsection{Proof of Lemma~\ref{lem:lyapLike}}
\label{sec:proofLemLyapLike}

To the end of proving Lemma~\ref{lem:lyapLike}, we note that model~\eqref{eq:xDot} and function~\eqref{eq:lyapLike} suggest that there are three relevant affine systems and smooth functions associated to the three cases in~\eqref{eq:vDot} that are worth considering (and will be used in our proofs). They correspond to
\begin{subequations}
\label{eq:subcasesAffine}
\begin{align}
\dot \xi &= f_1(\xi) := A \xi - b,   & & \xi(0)=\xi_1,  \label{eq:subcasesAffine:f1} \\
\dot \xi &= f_0(\xi) := \smallmat{0 & 0 & 0\\ 1 & 0 & 0 \\ 0 & 0 & 0} \xi , &  & \xi(0)=\xi_0,  \label{eq:subcasesAffine:f0} \\
\dot \xi &= f_{-1}(\xi) := A \xi + b, & &\xi(0)=\xi_{-1}, \label{eq:subcasesAffine:f-1}
\end{align}
and, with the definition $|\xi|^2_P := \xi^T P \xi$,
\begin{equation}
\label{eq:subcasesAffine:Vs}
V_1(\xi)\!:=\!\begin{vmatrix}\smallmat{\sigma\\ \phi-f_c\\v}\end{vmatrix} ^2_P\!, \,V_0(\xi)\!:=\!\begin{vmatrix}\smallmat{\sigma\\ 0 \\ 0}\end{vmatrix}^2_P\!,\,  V_{-1}(\xi)\!:=\!\begin{vmatrix}\smallmat{\sigma\\ \phi+f_c\\v}\end{vmatrix}^2_P.
\end{equation}
\end{subequations}

Based on the description above, we can state the following claim relating solutions to~\eqref{eq:xDot} and $V$ in~\eqref{eq:lyapLike} to~\eqref{eq:subcasesAffine}. Its proof mostly relies on straightforward inspection of the various cases and is given in Appendix~\ref{app:proofSuitAff}.

\begin{claim}
\label{claim:suitableAffine}
There exists $c>0$ such that, for each initial condition $(\bar \sigma,\bar \phi,\bar v)$, one can select $k\in\{-1,0,1\}$ and $T>0$ satisfying the following:
\begin{enumerate}
\item the unique solution $\xi=(\xi_\sigma,\xi_\phi,\xi_v)$ to the $k$-th initial value problem among \eqref{eq:subcasesAffine:f1}-\eqref{eq:subcasesAffine:f-1} with initial condition $\xi_k=(\bar \sigma,\bar \phi,\bar v)$ coincides in $[0,T]$ with the unique solution to~\eqref{eq:xDot}; \label{claim:suitableAffine:suitableK}
\item the solution $\xi$ mentioned above satisfies for all $t\in[0,T]$ \label{claim:suitableAffine:dVdt}
\begin{equation}
\label{eq:dVdt}
V(\xi(t))= V_k(\xi(t)), \quad \tfrac{d}{dt} V_k(\xi(t)) \le - c |\xi_v(t)|^2.
\end{equation}
\end{enumerate}
\end{claim}

\medskip

Additionally, we restate a fact from~\cite{hagood2006recovering} that is beneficial to proving Lemma~\ref{lem:lyapLike}. Specifically, we use \cite[Theorem~9]{hagood2006recovering} together with the variant in \cite[Section~5 (point~a.)]{hagood2006recovering}. We also specialize the statement, using the fact that when the function $g$ is integrable, the standard integral can replace the upper integral (as noted after \cite[Definition~8]{hagood2006recovering}). The lower right Dini derivative $\Di h$ of $h$ is defined as $\Di h(t):= \liminf_{\epsilon \to0^+} \frac{h(t+\epsilon)-h(t)}{\epsilon}$.

\begin{fact}
\label{fact:bound}
\cite{hagood2006recovering} Given $t_2 > t_1 \ge 0$, suppose that $h$ satisfies $\liminf_{\tau\to \bar \tau}h(\tau) \ge h(\bar \tau)$ (i.e., $h$ is lower semicontinuous) and that $l$ is locally integrable in $[t_1,t_2]$. If $\Di h(\tau) \le l(\tau)$ for all $\tau\in[t_1,t_2]$, then
\begin{equation*}
h(t_2)-h(t_1) \le \int_{t_1}^{t_2} l(\tau) d\tau.
\end{equation*}
\end{fact}

Building on Claim~\ref{claim:suitableAffine} and Fact~\ref{fact:bound} we can prove Lemma~\ref{lem:lyapLike}.

\begin{proofof}{\em Lemma~\ref{lem:lyapLike}.} We show first that $V$ is lower semicontinuous. Define the set-valued mapping 
\[
G(x)\!:=\!\!\!\!\!\!\!\! \bigcup_{f\in\SGN(v)} \!\!\!\!\!\!\! g(\sigma,\phi,v,f), \,\,\, g(\sigma,\phi,v,f)\!:=\!\smallmat{
\sigma\\
\phi-f\\
v}^T\!\! P \smallmat{
\sigma\\
\phi-f\\
v},
\]
and consider the additional set-valued mapping $(\sigma,\phi,v)\rightrightarrows H(\sigma,\phi,v):=(\sigma,\phi,v,f_c\SGN(v))$. By the very definition of set-valued mapping, we can write $G=g\circ H$ (the composition of $g$ and $H$), that is, $(\sigma,\phi,v)\rightrightarrows g(\sigma,\phi,v,f_c\SGN(v))=G(x)$. Then, $G$ is outer semicontinuous (osc) by \cite[Proposition~5.52, item~(b)]{rockafellar2009variational} because both $g$ and $H$ are osc and $H$ is locally bounded. Finally, by the definition of distance $d(u,S)$ between a point $u$ and a closed set $S$, we can write $V(x)=d(0,G(x))$. Then, $V$ is lsc by \cite[Proposition~5.11, item~(a)]{rockafellar2009variational} because $G$ was proven to be osc. 

We prove now the properties of $V$ item by item.

\textit{Item 1).} There exists $\mathfrak{g}>0$ such that $\smallmat{\sigma\\ v}^T \smallmat{\frac{k_v}{k_i} & -1\\ -1 & k_p} \smallmat{\sigma\\ v} \ge \mathfrak{g} (\sigma^2 + v^2)$  because the inner matrix is positive definite by Assumption~\ref{ass:PIDpars}. Moreover, from~\eqref{eq:lyapLikeV}, $\min\limits_{f\in f_c \SGN(v)} \big( \phi-f \big)^2 \ge \min\limits_{f\in [-f_c,f_c]} \big( \phi-f \big)^2=\dz(\phi)^2$. Therefore, \eqref{eq:distFromA} yields $V(x)\ge c_1 |x|_\A^2$ with $c_1:=\min\{\mathfrak{g},1\}$.

\textit{Item 2).} Equation~\eqref{eq:Vdecrease} is a mere application of Fact~\ref{fact:bound} for $h(\cdot)=V(x(\cdot))$ and $l(\cdot)=-c (v(\cdot))^2$ where $x=(\sigma,\phi,v)£$ is a solution to~\eqref{eq:xDot}. So, we need to check that the assumptions of Fact~\ref{fact:bound} are verified.

We already established above that $V$ is lsc. Solutions $x$ to~\eqref{eq:xDot} are absolutely continuous functions by definition.
Then, because the composition of a lower semicontinuous and a continuous function is lower semicontinuous (see \cite[Exercise~1.40]{rockafellar2009variational}), the Lyapunov-like function~\eqref{eq:lyapLikeV} evaluated along the solutions of~\eqref{eq:xDot} is lsc. Since solutions are absolutely continuous, $-c v^2$ is locally integrable.

Finally, it was proven in Claim~\ref{claim:suitableAffine}, item~\ref{claim:suitableAffine:suitableK} that for each initial condition, the unique solution to~\eqref{eq:xDot} coincides with the solution to one of the three affine systems in~\eqref{eq:subcasesAffine} (numbered $k$) on a finite time interval $T$. Moreover, from Claim~\ref{claim:suitableAffine}, item~\ref{claim:suitableAffine:dVdt} $V$ coincides in $[0,T]$ with the function $V_k$ in \eqref{eq:dVdt}, which is differentiable, then $V(x(t))$ at $t=0$ is at least differentiable from the right and the lower right Dini derivative coincides with the right derivative. In particular, we established in~\eqref{eq:dVdt} that this right derivative is upper bounded by $-c v^2$.
\end{proofof}

\section{Proof of Prop.~\ref{prop:GA+S}: stability}
\label{sec:stab}

\begin{figure}[!t]
\centering
\includegraphics[width=2.2in]{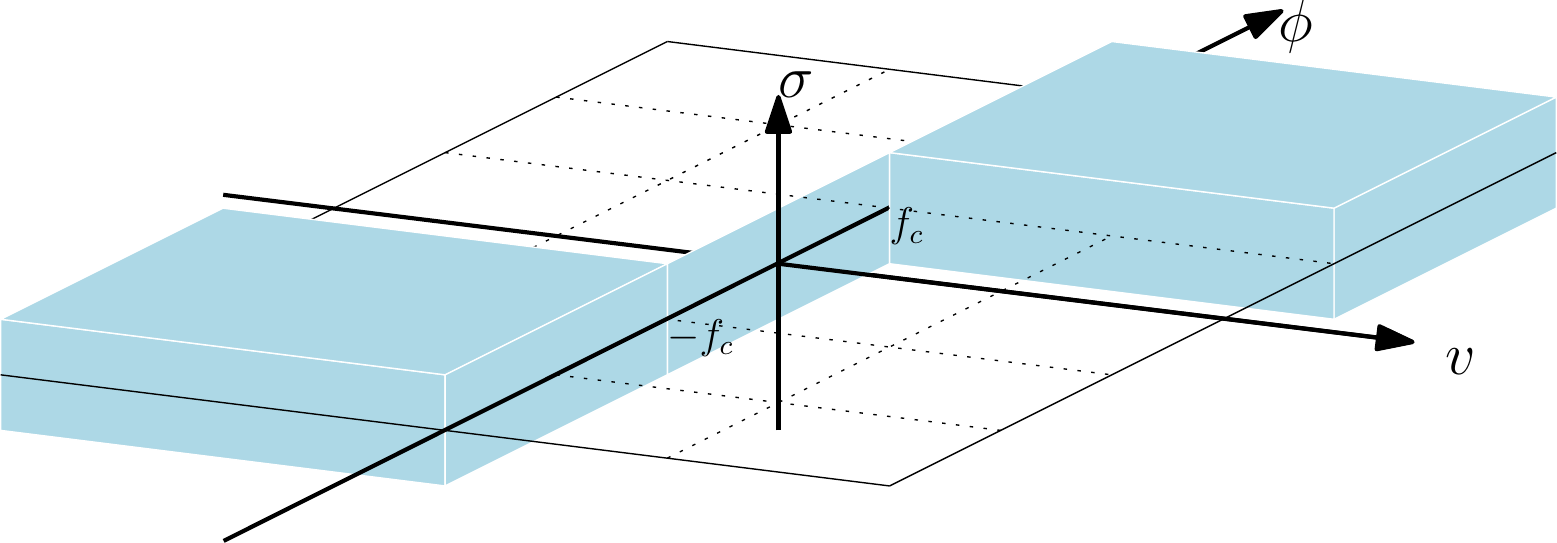}
\caption{$R$ is the (closed) blue region in Lemma~\ref{lem:lyapStab}, $\hat R$ is its complement.}
\label{fig:whiteBlue}
\end{figure}

The Lyapunov-like function introduced in~\eqref{eq:lyapLike} of the previous section is unfortunately not enough to prove stability. Indeed, its discontinuity on the attractor $\A$ prevents us from obtaining a uniform continuous upper bound depending on the distance from $\A$. 
However, a stability bound can be constructed through an auxiliary function defined as
\begin{equation}
\label{eq:Vhat}
\hat V(x)  := \tfrac{1}{2} k_1 \sigma^2 + \tfrac{1}{2} k_2 \big(\dz(\phi)\big)^2+k_3 |\sigma| |v| + \tfrac{1}{2} k_4 v^2.
\end{equation}
Function $\hat V$ allows establishing bounds in the directions of discontinuity of $V$. In particular, we define the two subsets
\begin{equation*}
\begin{split}
R &:=\{x \colon v(\phi -\sign(v)f_c ) \ge 0\} \\
\hat R &:=\real^3\backslash R\\
\end{split}
\end{equation*}
represented in Figure~\ref{fig:whiteBlue}. The following holds.
\begin{lemma}
\label{lem:lyapStab}
For suitable positive scalars $k_1,\dots,k_4$ in~\eqref{eq:Vhat}, there exist positive scalars $c_1$, $c_2$, $\hat c_1$, $\hat c_2$ such that
\begin{subequations}
\begin{align}
& c_1 |x|_\A^2 \le V(x) \le c_2 |x|_\A^2,\, &{ \forall } x\in R, \label{eq:VandR}\\
& \hat c_1 |x|_\A^2 \le \hat V(x) \le \hat c_2 |x|_\A^2,\, &{ \forall } x\in \hat R, \label{eq:hatVandHatR}\\
&  {\hat V}^\circ(x)  := \max_{\mathfrak{v} \in\partial\hat V(x),\mathfrak{f}\in F(x)} \langle \mathfrak{v},\mathfrak{f} \rangle  \le 0,\, &{ \forall } x \in \hat R,\label{eq:hatVdot}
\end{align}
\end{subequations}
where $\partial \hat V(x)$ denotes the generalized gradient of $\hat V$ at $x$ (see~\cite[\S1.2]{clarke1990optimization}) and $F$ is the set-valued map in~\eqref{eq:xDot}.
\end{lemma}
\begin{proof}
\begin{subequations}
Note that $\min\limits_{f\in f_c \SGN(v)} \big( \phi-f \big)^2 = \dz (\phi)^2$ whenever $x\in R$. Since $P$ in~\eqref{eq:lyapLikeP} is positive definite and $V(x)=\smallmat{\sigma\\ \dz(\phi)\\ v}^T P\smallmat{\sigma\\ \dz(\phi)\\ v}$ in $R$, positive $c_1$ and $c_2$ can be chosen to satisfy \eqref{eq:VandR}, using the definition~\eqref{eq:distFromA}. (The lower bound in~\eqref{eq:VandR} was already established for all $x\in\real^3$ in Lemma~\ref{lem:lyapLike}, item~\ref{lem:lyapLike:item:lowerbound}.) For positive $k_1,\dots,k_4$ and $k_1 k_4 > k_3^2$, the inner matrix in $\hat V(x)=\frac{1}{2} \smallmat{|\sigma|\\ |\dz(\phi)| \\ |v| }^T \smallmat{{k_1}  &  0 & {k_3}\\ 0 & {k_2} & 0 \\ {k_3}  &  0 & {k_4}} \smallmat{|\sigma|\\ |\dz(\phi)| \\ |v|}$ is positive definite and \eqref{eq:hatVandHatR} can be satisfied for the same reason.

To prove~\eqref{eq:hatVdot}, we consider only the set  $\hat R_>:=\hat R \cap \{x \colon v>0\}$ because a parallel reasoning can be followed in $\hat R \cap \{x \colon v<0\}$. For $x\in\hat R_>$, we have $v>0$, $\phi<f_c$ and \eqref{eq:xDot}~reduces to the differential equation
\begin{equation}
\label{eq:diffEqForHatV}
\begin{aligned}
\dot \sigma & = - k_i v =: \mathfrak{f}_\sigma(x) \\
\dot \phi & = \sigma - k_p v =: \mathfrak{f}_\phi(x) \\
\dot v & = - k_v v + \phi - f_c =: \mathfrak{f}_v(x) \le - k_v |v| -|\dz(\phi)|. \\
\end{aligned}
\end{equation}
Consistently, the $\max$ in~\eqref{eq:hatVdot} is to be checked only for the singleton $\mathfrak{f} = (\mathfrak{f}_\sigma(x),\mathfrak{f}_\phi(x),\mathfrak{f}_v(x))$ which $F(x)$ reduces to for all $x\in\hat R_>$. Moreover,
\begin{equation}
\label{eq:almostEverywhereDerivatives}
\tfrac{d}{d\phi}\left(\tfrac{1}{2} \big(\dz(\phi)\big)^2 \right)=\dz(\phi),\,\,\partial\left({|\sigma|}\right)= \SGN(\sigma),
\end{equation}
where $\partial\left({|\sigma|}\right)$ denotes the generalized gradient of $\sigma \mapsto |\sigma|$. We need then to find suitable positive constants $k_1,\dots,k_4$ satisfying $k_1 k_4 > k_3^2$ and such that ${\hat V}^\circ(x)$ is negative semidefinite in $\hat R_>$. 
Since in $\hat R_>$we have $v=|v|$ and $\dz(\phi)=-|\dz(\phi)|$, then we get 
$\max_{\zeta \in \partial|\sigma|}(-k_i k_3|v|^2 \zeta)=k_i k_3 |v|^2$ for all $x \in \hat R_>$, which gives in turn:
\begin{equation*}
\begin{aligned}
{\hat V}^\circ(x) \!   \le \!  [k_i k_3 |v|^2\!\!  -k_4 k_v |v|^2]\!  + \![k_2 \sigma \dz(\phi)\! - k_3 |\sigma| |\dz(\phi)| ]\\
+[ -k_1 k_i \sigma v \!- k_3 k_v | v | |\sigma| ] \! + \![k_2 k_p |v | |\dz(\phi) | \!-\!k_4 |v| |\dz(\phi)| ]. \\
\end{aligned}
\end{equation*}
Since $k_1,\dots, k_4$ are positive by assumption, in each pair in brackets the second term is negative semidefinite and dominates the first (sign-indefinite or nonnegative) term as long as $k_3>\max\big\{ \frac{k_i}{k_v}k_1,k_2\big\}$ and $k_4 > \max\big\{\frac{k_i}{k_v}k_3,k_p k_2,\frac{k_3^2}{k_1}\big\}$. With this selection, \eqref{eq:hatVandHatR} and \eqref{eq:hatVdot} are simultaneously satisfied.
\end{subequations}
\end{proof}

\begin{proofof}{\em item 2) of Proposition~\ref{prop:GA+S} (stability).}
\begin{subequations}
Based on the constants $c_1$, $c_2$, $\hat c_1$, $\hat c_2$ introduced in Lemma~\ref{lem:lyapStab}, the following stability bound for each solution $x$ to~\eqref{eq:xDot}
\begin{equation}
\label{eq:stabState}
|x(t)|_\A \le \sqrt{\tfrac{c_2\hat c_2}{c_1\hat c_1}} |x(0)|_\A, \,\, \forall t \ge 0
\end{equation}
is proven by splitting the analysis in two cases.\newline
\emph{Case~(i):} $x(t) \notin R$, $\forall t \ge 0$. Since $R \cup \hat R = \real^3$, $x(t)\in \hat R$ for all $t\ge 0$ and from~\eqref{eq:hatVdot} 
\begin{equation*}
\hat V^\circ(x(t)) \le 0,\, \forall t \ge 0 \Rightarrow \hat V(x(t)) \le \hat V(x(0)),\, \forall t \ge 0.
\end{equation*}
Using bound~\eqref{eq:hatVandHatR} we obtain
\begin{equation*}
\hat c_1 |x(t)|_\A^2 \le \hat V(x(t)) \le \hat V(x(0)) \le \hat c_2 |x(0)|_\A^2,\,\, \forall t \ge 0,
\end{equation*}
which implies~\eqref{eq:stabState} because $1 \le \sqrt{{c_2}/{c_1}}$  from~\eqref{eq:VandR}.\newline
\emph{Case~(ii):} $\exists t_1 \ge 0$ such that $x(t_1) \in R$.  Consider the smallest $t_1\ge 0$ such that $x(t_1)\in R$ (the existence of such a \emph{smallest} time follows from $R$ being closed). Then, following the analysis of Case~(i) for the (possibly empty) time interval $[0,t_1)$ and using continuity of solutions, we obtain
\begin{equation}
\label{eq:boundInHatR}
\hat c_1 |x(t)|^2_\A \le \hat c_2 |x(0)|^2_\A, \,\, \forall t\in [0,t_1].
\end{equation}
At $t_1$ we apply \eqref{eq:VandR} (because $x(t_1)\in R$) and \eqref{eq:boundInHatR} to obtain $V(x(t_1))\le c_2\big( \frac{\hat c_2}{{\hat c_1}} |x(0)|^2_\A\big)$. Finally, by the bounds in items~\ref{lem:lyapLike:item:lowerbound} and \ref{lem:lyapLike:item:decrease} of Lemma~\ref{lem:lyapLike},
\begin{equation}
\label{eq:boundInHatRandInR}
c_1 |x(t)|_\A^2 \le V(x(t)) \le V(x(t_1)) \le  c_2 \tfrac{\hat c_2}{{\hat c_1}} |x(0)|^2_\A,\, \forall t\ge t_1.
\end{equation}
Since $\sqrt{\tfrac{c_2}{c_1}}\ge1$, \eqref{eq:boundInHatR} implies $c_1 |x(t)|_\A^2 \le  c_2 \tfrac{\hat c_2}{{\hat c_1}} |x(0)|^2_\A,\, \forall t\in [0,t_1]$, which proves \eqref{eq:stabState} when combined with~\eqref{eq:boundInHatRandInR}.
\end{subequations}
\end{proofof}

\section{Conclusions and future work}
In this work we characterized properties of a differential inclusion model of the feedback interconnection of a sliding mass with a PID controller under Coulomb friction. We proved global asymptotic stability of the largest set of closed-loop equilibria. Due to the regularity properties of the differential inclusion model, global asymptotic stability is intrinsically robust. Additionally, taking as input the size of the inflation of a perturbed model, the dynamics is input-to-state stable, and this perturbation includes the well-known Stribeck effect. Future work will address further the case of static friction force larger than the Coulomb one and will propose for that setting compensation schemes relying on the proposed Lyapunov-based proof.

\bibliographystyle{IEEEtranS}
\bibliography{IEEEabrv,refs}

\appendices
\section{Proof of Claim~\ref{claim:suitableAffine}}
\label{app:proofSuitAff}

\begin{proofof}{\em Claim~\ref{claim:suitableAffine}.}
The proof of the claim shows for each possible initial conditions that for a suitable $k$ the solution to the affine system $\dot \xi = f_k(\xi)$ among~\eqref{eq:subcasesAffine:f1}-\eqref{eq:subcasesAffine:f-1} is also solution to~\eqref{eq:xDot} on the interval $[0,T]$. By Lemma~\ref{lem:uniq}, this is also the unique solution to~\eqref{eq:xDot} on the same interval. The suitable $k$ for each initial condition is listed in the next table. 
\begin{equation*}
\begin{array}{llr}
\toprule
\text{Case} & \text{Possible initial condition} & k\\
\midrule
\text{(i)} & \bar v > 0 & 1\\
\text{(ii)} & \bar v=0, \bar \phi > f_c & 1\\
\text{(iii)} & \bar v = 0, \bar \phi = f_c, \bar \sigma >0 & 1\\
\text{(iv)} & \bar v = 0, \bar \phi = f_c, \bar \sigma = 0 & 0\\
\text{(v)} & \bar v = 0, \bar \phi = f_c, \bar \sigma < 0 & 0\\
\text{(vi)} & \bar v = 0, |\bar \phi| < f_c & 0\\
\text{(vii)} & \bar v = 0, \bar \phi = - f_c, \bar \sigma > 0 & 0\\
\text{(viii)} & \bar v = 0, \bar \phi = - f_c, \bar \sigma = 0 & 0\\
\text{(ix)} & \bar v = 0, \bar \phi = - f_c, \bar \sigma < 0 & -1\\
\text{(x)} & \bar v = 0, \bar \phi < - f_c & -1\\
\text{(xi)} & \bar v < 0 & -1\\
\bottomrule
\end{array}
\end{equation*}
Define the components of the solution $\xi$ as $(\xi_\sigma,\xi_\phi,\xi_v)$. We first prove item 1) and the left equation in~\eqref{eq:subcasesAffine:Vs}, addressing only cases (i)-(v) as the remaining ones follow from parallel arguments.

\textit{Case}~(i). We choose $k=1$, that is, \eqref{eq:subcasesAffine:f1}. Because $\xi_v(0)=\bar v>0$, there exists a finite time $T$ such that for all $t\in [0,T]$ $\xi_v(t)>0$. When the resulting solution $\xi$ is substituted into~\eqref{eq:xDot}, $-f_c \SGN(\xi_v(t))=\{-f_c\}$ for all $t\in [0,T]$ so that for all $t\in [0,T]$ \eqref{eq:xDot} becomes $\dot \xi(t) = A \xi(t) - b$, and this is satisfied precisely because $\xi$ arises from~\eqref{eq:subcasesAffine:f1} ($k=1$). Then the solution $\xi$ is also a solution to~\eqref{eq:xDot} for $t\in [0,T]$ because they have the same initial conditions and $\dot \xi(t) \in F(\xi(t))$. $V(\xi(t))=V_1(\xi(t))$ for all $t\in[0,T]$ because $f=f_c$ is the only possible selection in~\eqref{eq:lyapLikeV} due to $\xi_v(t)>0$ in $[0,T]$.

\textit{Case}~(ii). We prove that from~\eqref{eq:subcasesAffine:f1} (that is, $k=1$) and initial conditions $\xi(0)=(\bar \sigma,\bar \phi,0)$ with $\bar \phi>f_c$ it is implied that there exist $T>0$ such that for all $t\in(0,T]$, $\xi_v(t)>0$. Then, as in Case~(i), we can conclude that $\xi$ is also a solution to~\eqref{eq:xDot}. Indeed, the third state equation of~\eqref{eq:subcasesAffine:f1} reads $\dot \xi_v = \xi_\phi - k_v \xi_v - f_c$ with $\xi_v(0)=0$, $\xi_\phi(0)=\bar \phi>f_c$, so that $\dot \xi_v (0)= \xi_\phi(0)- f_c>0$. Also here $V(\xi(t))=V_1(\xi(t))$ because $\xi_v(t)>0$ for all $t\in(0,T]$ and $\xi_\phi(0)>f_c$ so that the minimizer in~\eqref{eq:lyapLikeV} at $t=0$ is also $f=f_c$.

\textit{Case}~(iii). We prove that from~\eqref{eq:subcasesAffine:f1} ($k=1$) and initial conditions $\xi(0)=(\bar \sigma,f_c,0)$ with $\bar \sigma>0$ it is implied that there exist $T>0$ such that for all $t\in(0,T]$, $\xi_v(t)>0$. Then, as in Case~(i), we can conclude that $\xi$ is also a solution to~\eqref{eq:xDot}. Indeed, the second state equation of~\eqref{eq:subcasesAffine:f1} reads in this case $\dot \xi_\phi = \xi_\sigma - k_p \xi_v$ and $\dot \xi_\phi (0)=\xi_\sigma(0)>0$, so that there exists $T$ such that $\phi(t) > f_c$ for all $t\in(0,T]$. The third state equation of~\eqref{eq:subcasesAffine:f1} reads $\dot \xi_v = \xi_\phi - k_v \xi_v -f_c$ and is $\dot \xi_v(0) = 0$ when evaluated at time $0$. Therefore, to show that $\xi_v(t)>0$ for all $t\in(0,T]$, we need to differentiate the third equation and get $\ddot \xi_v(0) = \bar \sigma>0$ and by this we conclude that $\dot \xi_v(t)>0$ in $(0,T]$ and then also $\xi_v(t)>0$ in $(0,T]$. A similar reasoning to the previous case shows that $V(\xi(t))=V_1(\xi(t))$, for all $t\in[0,T]$.

\textit{Case}~(iv). By choosing $k=0$, we can verify that the solution to~\eqref{eq:xDot} and \eqref{eq:subcasesAffine:f0} is constant (specifically, an equilibrium). Then $V(\xi(t))=V_0(\xi(t))$ for all $t\in[0,+\infty)$ because $V(\xi(0))=V_0(\xi(0))$.

\textit{Case}~(v). Here we have the explicit solution $\xi_\sigma(t)=\bar \sigma<0$, $\xi_\phi(t) = f_c+\bar \sigma t$, $\xi_v(t)=0$, for the interval $[0,T]=[0,-{2f_c}/{\bar \sigma}]$. This solution satisfies \eqref{eq:xDot} because an appropriate selection of the value in $\SGN(0)$ obtains $\xi_\phi(t) -f_c\SGN(\xi_v(t))=\xi_\phi(t) -f_c \SGN(0)=\{0\}$ for all $t\in[0,-{2f_c}/{\bar \sigma}]$, as in this interval $|\xi_\phi(t)|\le f_c$.
In this case the minimizer in~\eqref{eq:lyapLikeV} is $f(t)=\xi_\phi(t)$ for all $t\in[0,T]$, so that $V(\xi(t))=V_0(\xi(t))$ for all $t\in[0,T]$.

Finally, $c=2( k_v k_p -k_i)>0$ (by Assumption~\ref{ass:PIDpars}) is such that the right inequality in~\eqref{eq:dVdt} holds for each $k\in\{-1,0,1\}$ regardless of the initial condition and case under consideration. For $k=1$ write $\frac{d}{dt} V_1(\xi(t))=\frac{d}{dt}\bigg(\smallmat{\xi_\sigma\\ \xi_\phi-f_c\\ \xi_v}^T P \smallmat{\xi_\sigma\\ \xi_\phi-f_c\\ \xi_v}\bigg)=(A\xi -b)^T P \smallmat{\xi_\sigma\\ \xi_\phi-f_c\\ \xi_v} + \smallmat{\xi_\sigma\\ \xi_\phi-f_c\\ \xi_v}^T P (A \xi - b)=- c \xi_v^2$, that satisfies \eqref{eq:dVdt} in the interval $[0,T]$. Parallel computations hold for $k=-1$. For $k=0$, $V_0(\xi(t))=\frac{k_v}{k_i}\xi_v(t)^2$ so that $\frac{d}{dt} V_0(\xi(t))=2\frac{k_v}{k_i}\xi_\sigma \dot{\xi}_\sigma=0\le 0 = c \xi_v(t)^2$.
\end{proofof}

\end{document}